\documentclass[12]{article}

\usepackage{amsmath,amsthm,amssymb}
\usepackage{fullpage}
\usepackage{graphics,graphicx}
\usepackage{mwe}
\usepackage{caption}
\usepackage{subcaption}
\usepackage{booktabs}
\usepackage{setspace}
\usepackage{natbib}     

\usepackage{hyperref}

\newtheorem{definition}{Definition}
\newtheorem{theorem}{Theorem}

\newtheorem{lemma}{Lemma}
\newtheorem{proposition}{Proposition}

\hypersetup{
    colorlinks=true,
    linkcolor=blue,
    filecolor=magenta,      
    urlcolor=blue,
    citecolor=blue,     
    pdfpagemode=FullScreen,
}

\title{Segregation Monotonicity and the Measurement of Inequality in Social Networks}

\author{Deepankar Basu\thanks{Department of Economics, University of Massachusetts Amherst. Email: \url{dbasu@umass.edu}. I would like to thank Debarshi Das, Arjun Jayadev, Debraj Ray and Rajiv Sethi for  helpful comments on previous versions of this paper. All remaining errors are mine.}}

\begin{document}

\maketitle

\doublespacing

\begin{abstract}
This paper introduces segregation monotonicity as a criterion for evaluating measures of inequality in social networks. A network inequality measure satisfies segregation monotonicity if, holding the distribution of income fixed, it weakly increases as the network becomes more segregated according to a specified transformation of network architecture. I investigate this property using a class of level-$k$ star networks that represent increasing social segregation in the following sense: relatively poorer individuals become increasingly isolated from one another and their social comparisons become increasingly concentrated among richer individuals. I show that a relative deprivation-based measure of inequality, which aggregates comparisons with richer network neighbors, satisfies segregation monotonicity. By contrast, total experience-based measures, which aggregate absolute income differences among all network neighbors, do not generally satisfy segregation monotonicity and can decline as segregation rises. I also establish several relationships between the total experience-based measures and the standard Gini coefficient. The results show that alternative network-based inequality measures can embody fundamentally different conceptions of socially relevant comparisons.\\
Keywords: Inequality measurement; social networks; segregation monotonicity; relative deprivation; experienced inequality.\\
JEL Codes: D63, C43, Z13.
\end{abstract}

\section{Introduction}
Economic inequality is usually understood as a property of a distribution of income, wealth, or some other economically relevant variable. This perspective abstracts from the social structure through which individuals observe, experience, and compare economic differences. Yet, as a large literature in social psychology emphasize, individuals do not generally observe or compare themselves with the entire population. Rather, their perceptions of economic differences are shaped by the particular social environments in which they are embedded and by the individuals with whom they interact \citep{festinger_1954, dimaggio_garip_2012, chiang_2011, chiang_2015}. Once social structure is taken seriously, inequality becomes a property not only of the distribution of income but also of the architecture of socially relevant comparisons.

This observation raises a fundamental question for the measurement of inequality in social networks. Suppose that the distribution of income remains unchanged, but the network of social connections becomes increasingly segregated. In particular, suppose that relatively poorer individuals become increasingly isolated from one another while their remaining social comparisons become increasingly concentrated among richer individuals. Should a measure of inequality register such a change in network architecture as an increase in inequality?

This paper argues that this question provides a useful criterion for evaluating measures of inequality in social networks. I call this criterion \textit{segregation monotonicity}. Informally, a network inequality measure satisfies segregation monotonicity if, holding the distribution of income fixed, it weakly increases when the network undergoes a specified transformation that increases social segregation. The criterion directs attention to a feature of network inequality measures that is distinct from their values for a particular distribution and network architecture: how they respond to meaningful \textit{changes} in the structure of social comparisons.

The motivation for this criterion comes from a growing literature on social comparison, networks, and inequality. Research in social psychology and economics suggests that individuals do not form judgments about economic inequality by comparing themselves with the entire distribution of income in society \citep{festinger_1954, chiang_2015}. Rather, the relevant reference groups are shaped by local social environments. Experimental and theoretical work has shown that judgments of distributional inequality can depend on the structure of social networks and on the composition of individuals' local reference groups \citep{chiang_2011}. Related research has demonstrated that network topology can substantially affect assessments of inequality even when the underlying distribution of income remains unchanged \citep{chiang_2015}.

The importance of network architecture is also emphasized in the broader literature on social networks and social inequality. Network effects can amplify initial economic differences through mechanisms such as homophily, cumulative advantage, differential access to information, and unequal opportunities for social interaction \citep{dimaggio_garip_2012}. From this perspective, inequality is not only reflected in the distribution of economic resources but can also be reinforced by inequalities in the structure of social connections. Empirical research has further linked fragmented and segregated social networks with higher levels of economic inequality \citep{Toth2021}. These findings suggest that changes in social segregation may themselves be relevant to the measurement of inequality, even when individual incomes remain unchanged.

The concept of segregation monotonicity does not require a universal ordering of arbitrary networks according to their degree of segregation. Such an ordering would itself raise difficult conceptual and empirical questions. Instead, segregation monotonicity is defined with respect to a specified and economically meaningful transformation of network architecture. The question is therefore whether a measure responds monotonically when a network undergoes a particular transformation that represents an increase in segregation while holding the distribution of income fixed.

To investigate this question, I introduce a class of \textit{level-$k$ star networks}. A standard star network, in which a single individual occupies the center, is generalized to a structure with $k$ ($1 \leq k \leq n-1$) individuals at a shared center. Moving from a level-$k$ star network to a level-$(k-1)$ star network represents a particular increase in social segregation: relatively poorer individuals become increasingly isolated from one another, while their remaining social comparisons become increasingly concentrated among richer individuals.

This construction is related to ideas developed in the literature on network inequality. For example, \citet{BOWLES2020108789} emphasize that richer individuals often have more social connections. It stands to reason that richer individuals may therefore derive relatively greater benefits from social connections, for instance through social capital or information flows related to economic opportunities. As social networks become increasingly concentrated around richer individuals, the relatively poorer individuals are impoverished in terms of social connections. Turning to comparisons that are relevant for experienced inequality, we see that the average income level of the comparison groups of poorer individuals may rise. Consequently, poorer individuals may experience greater inequality even when individual income levels remain unchanged \citep{BOWLES2020108789, chiang_2015}. While Bowles and Carlin illustrate this idea using a three-person example, the level-$k$ star construction developed in this paper generalizes it to an $n$-person setting and provides a precise transformation with which to evaluate segregation monotonicity.

The central results of the paper concern two recently proposed approaches to constructing measures of inequality in social networks. The first approach is what I call the \textit{total experience-based approach}. This approach was implemented, respectively, by \citet{BOWLES2020108789} and \citet{MAMUNURU2025106799}. These measures aggregate absolute income differences among all socially connected individuals. The second approach is what I call a \textit{relative deprivation-based approach}. This approach, implemented in \citet{stark_etal_2024}, considers only upward comparisons with richer network neighbors while constructing a measure of inequality.

The inequality measures proposed in these two approaches differ sharply in their response to increasing segregation. Holding individual income levels fixed, I show that the relative deprivation-based measure, denoted by $I^{SBF}$ (SBF = Stark, Bielawski and Falniowski), weakly increases when society moves from a level-$k$ to a level-$(k-1)$ star network. Thus, $I^{SBF}$ satisfies segregation monotonicity with respect to this network transformation. By contrast, the total experience-based measures, denoted by $I^{BC}$ (BC = Bowles and Carlin) and $I^{MSJ}$ (MSJ = Mamunuru, Shrivastava and Jayadev), do not generally satisfy segregation monotonicity. I construct a counterexample in which increasing segregation causes both measures to decline.

These results show that alternative measures of inequality in social networks can embody fundamentally different conceptions of which social comparisons matter. The relative deprivation-based approach gives special significance to upward comparisons and therefore responds systematically when poorer individuals' comparisons become increasingly concentrated among richer individuals. The total experience-based approach aggregates absolute income differences symmetrically across network connections and need not respond monotonically to the same increase in segregation. The distinction between the approaches is therefore not merely a matter of alternative formulas for aggregating the same network-based income differences. It concerns the more fundamental question of how changes in the architecture of social comparison should affect the measurement of inequality.

The paper also establishes several additional relationships among the measures. For general populations, when network structures are not taken into account, the Gini coefficient ($G$) is a widely used measure of inequality in distributions of income, wealth, or other relevant variables.\footnote{In this paper, I use income as the relevant variable.} Its popularity rests on several desirable properties \citep{ray_1998, lambert_2001}. The total experience-based measures developed by \citet{BOWLES2020108789} and \citet{MAMUNURU2025106799} take the Gini coefficient as their point of departure. I show that $I^{BC}=I^{MSJ}$ in both a complete network and a star network; that $ 0 \leq I^{BC}=I^{MSJ} \leq 1$ in a complete network; and that $ I^{BC}=I^{MSJ}=G/(1-(1/n))$ in a complete network with $n$ individuals. These results clarify the relationship between the total experience-based measures and the standard Gini coefficient and provide a useful basis for the subsequent comparison of their behavior under increasing segregation.

The remainder of the paper is organized as follows. Section~\ref{sec:bc-formulation} discusses the two approaches for constructing measures of inequality in networks. Section~\ref{sec:relationship} compares the measures coming out of the two approaches in terms of some key properties. Section~\ref{sec:star-k} introduces level-$k$ star networks and analyzes segregation monotonicity. Section~\ref{sec:conclusion} concludes. While shorter proofs are part of the main text, longer proofs are collected together in an appendix.

\section{Two Approaches}\label{sec:bc-formulation}
In this section, I discuss two approaches to measuring inequality in networks. The first approach, developed in \citet{BOWLES2020108789} and \citet{MAMUNURU2025106799}, is what I will call the total experience-based approach. It takes the standard interpretation of the Gini coefficient as the point of departure and offer an alternative based on its critique of the standard Gini coefficient. The second approach, developed in \citet{stark_etal_2024}, is what I will call the relative deprivation-based approach. It takes as its point of departure the idea of relative deprivation that was presented by \citet{yitzhaki_1979} as a way to reinterpret the Gini coefficient. 

\subsection{Total experience-based approach}
Let us start by recalling the details of the standard Gini coefficient.\footnote{\citet{ray_2024} points out that Corrado Gini proposed $13$ measures of inequality. The two that have survived the test of time are what I will refer to in this paper as $G$ (the standard Gini coefficient) and $I^{BC}$ (the Bowles-Carlin reformulated Gini coefficient).} Consider a population of $n$ individuals with non-negative income levels $y_i \geq 0$, where $i=1, 2, \ldots, n$, where at least one person has strictly positive income. Let $\bar{y}=n^{-1} \sum_{i=1}^{i=n} y_i$ denote the mean income level in this population, which is strictly positive because at least one person has strictly positive income. The standard definition of the Gini coefficient \citep[page~31]{sen_1997} is
\begin{equation}\label{gini-std}
G = \frac{1}{2 n^2 \bar{y}} \sum_{j=1}^{j=n} \sum_{i=1}^{i=n} \big| y_j-y_i \big| =  \left[ \frac{1}{n^2} \sum_{j=1}^{j=n} \sum_{i=1}^{i=n} \big| y_j-y_i \big|  \right] \frac{1}{\bar{y}} \frac{1}{2},
\end{equation} 
which, as re-written on the far right hand side, shows it to be the product of three factors: (a) the arithmetic mean of absolute differences of income (first term), (b) the reciprocal of mean income (second term), and (c) one-half (third term). 

In the standard definition of the Gini coefficient in (\ref{gini-std}), $n$ of the terms in the sum are identically zero. These are the absolute values of the `own comparisons,' $\big| y_j-y_j \big|, j=1, 2, \ldots, n$. \citet{BOWLES2020108789} suggest removing such terms from the numerator and accordingly changing the denominator to reflect the reduced number of terms being added in the denominator. Accordingly, their reformulated Gini coefficient, $I^{BC}$, is given by
\begin{equation}\label{gini-bowles}
I^{BC} =  \left[ \frac{1}{\frac{n(n-1)}{2}} \sum_{j=1}^{j=n-1} \sum_{i=j+1}^{i=n} \big| y_j-y_i \big|  \right] \frac{1}{\bar{y}} \frac{1}{2}.
\end{equation} 
This is once again a product of the same three factors that appear in $G$ with the only difference being that the first factor, the arithmetic mean of absolute differences of income, is computed for the $n(n-1)/2$ pairs of unique non-identical comparisons, instead of $n^2$ (as in the standard Gini coefficient).\footnote{As pointed out by \citet{sethi_2024}, similar reformulations of the standard Gini coefficient have been previously explored by \citet{jasso_1979} and \citet{deaton_1997}.} According to \citet{BOWLES2020108789}, this index is a way to measure `experienced inequality' because it only takes account of unique non-identical comparisons (instead of all possible, including, own comparisons).\footnote{For an illuminating discussion about violation of the population (replication) principle axiom by $I^{BC}$, see \citet{ray_2024}, \citet{sethi_2024} and \citet{bowles_carlin_2024}.}   

The Bowles-Carlin intuition was further developed by \citet{MAMUNURU2025106799}. To conveniently express their measure of inequality, $I^{MSJ}$, we need $2$ symmetric $n \times n$ matrices, $A$ (adjacency matrix) and $D$ (absolute difference matrix). The $(i,j)$ element of the adjacency matrix, $A$, captures the adjacency of (or direct connection between) individuals $i$ and $j$, as follows:
\begin{equation}\label{adj-matrix}
a_{ij} = 
\begin{cases}
1 \quad \text{ if } (i,j) \text{ directly connected and } i \neq j  \\
0 \quad \text{ otherwise},
\end{cases}
\end{equation}
so that the nonzero elements of the $i$-th row of $A$ records all the direct connections of the $i$-th individual in the network. Note that $A$ is symmetric because $a_{ij}$ and $a_{ji}$ is equal by definition. Furthermore, the principal diagonal of $A$ has all zero entries. The $(i,j)$ element of the absolute difference matrix, $D$, is the absolute difference of incomes of individuals $i$ and $j$, as follows:
\begin{equation}\label{diff-matrix}
d_{ij} = 
\begin{cases}
|y_i - y_j|  \quad \text{ if } i \neq j  \\
0 \quad \text{ otherwise}.
\end{cases}
\end{equation}
Much like $A$, $D$ is symmetric because $|y_i - y_j|=|y_j - y_i|$, and it has zeros on the principal diagonal. The $i$-th row (or column) of $D$ records the absolute difference of individual $i$'s income with all other individuals in the network.
 
Using $A$ and $D$, we can now define, for each individual $i=1, 2, \dots, n$, what \citet[page~5--6]{MAMUNURU2025106799} call the \textit{neighborhood average difference} as, 
\begin{equation}\label{eq:delta-defn}
\Delta_i = \frac{A_{i *} D_{* i}}{A_{i *} \mathbf{1}_n} = \frac{\sum_{j = i}^{n} a_{ij} |y_j - y_i|}{\sum_{j=1}^{n} a_{ij}},
\end{equation}
where $A_{i *}$ and $ D_{* i}$ denote $i$-th row and $i$-th column of $A$ and $D$ respectively and $\mathbf{1}_n$ denotes a $n$-vector of $1$s. By using the $i$-th row of $A$, the neighborhood average difference computes the arithmetic mean of absolute differences for the $i$-th individual with respect to only those individuals who are directly connected to (or direct neighbors of) this individual. This is precisely how they formalize the notion of experienced inequality: comparisons occur only with direct neighbors, not with all individuals in the network (society).  

The measure of inequality in \citet{MAMUNURU2025106799} is then defined as,
\begin{equation}\label{gini-mjs}
I^{MSJ} = \left[ \frac{1}{n} \sum_{i=1}^{n} \Delta_i \right] \frac{1}{\bar{y}}\frac{1}{2}. 
\end{equation}
Once again, this is a product of the same three factors that appear in $G$ or $I^{BC}$ with the only difference being that the first factor is now the arithmetic mean of neighborhood average differences across all the $n$ individuals in the network. 

\subsection{Relative deprivation-based approach}
The measure of inequality for networks that was proposed in \citet{stark_etal_2024} builds on the idea of relative deprivation presented in \citet{yitzhaki_1979}. To understand this measure, we will continue using the adjacency matrix $A$ defined in (\ref{adj-matrix}). But instead of taking the absolute income difference of an individual's income with all her direct neighbors, as in the total experience-based approach presented above, we will restrict this difference to the subset of direct neighbors who are richer than the individual and identically set the difference to $0$ for poorer neighbors. Thus, we are only looking up the income ladder. It is this feature that captures the idea of relative deprivation: we are relatively deprived only with respect to those who are richer than us.

To conveniently express the measure of inequality proposed in \citet{stark_etal_2024}, which I will denote by $I^{SBF}$, we need another symmetric $n \times n$ matrices, $P$ (positive income difference matrix). The $(i,j)$ element of the positive difference matrix, $P$, gives the difference of incomes of individuals $i$ and $j$ only if that difference is positive, as follows:
\begin{equation}\label{pos-diff-matrix}
p_{ij} = 
\begin{cases}
\max \left\lbrace y_i - y_j, 0 \right\rbrace   \quad \text{ if } i \neq j  \\
0 \quad \text{ otherwise}.
\end{cases}
\end{equation}
Unlike $A$ and $D$, the matrix $P$ is not symmetric because $\max \left\lbrace y_i - y_j, 0 \right\rbrace \neq \max \left\lbrace y_j - y_i, 0 \right\rbrace$ when the nonnegative income levels are unequal, $y_i \neq y_j$. But like $A$ and $D$, the matrix $P$ has zeros on the principal diagonal. The nonzero elements of the $i$-th column of $P$ records the income shortfalls of individual $i$ with respect to all her richer counterparts in the network.

Using $A$ and $P$, we can define a measure of `relative deprivation', for individual $i$, as,
\begin{equation}
RD_i = \frac{A_{i *} P_{* i}}{A_{i *} \mathbf{1}_n} = \frac{\sum_{j = 1}^{n} a_{ij} \max \left\lbrace y_j - y_i, 0 \right\rbrace }{\sum_{j=1}^{n} a_{ij}},
\end{equation}
where I use the $i$-th column of $P$ because $p_{ji}$ captures the deprivation of $i$ relative to the richer individual $j$. By using the $i$-th row of $A$, once again, the comparison is restricted to direct neighbors only; and the use of the $i$-th column of $P$ ensures we only compute the arithmetic mean of income shortfalls.  

The measure of inequality proposed in \citet{stark_etal_2024}, is then obtained by normalizing the sum of $RD_i$ across all individuals by the maximum possible sum of $RD_i$ over all possible network structures keeping the \textit{total} income of the $n$ individuals in the network fixed while allowing the distribution of that fixed total income to change. The maximum possible sum of $RD_i$ arises in a star network where all the income is received by the person at the center (with zero income for the others) and can be shown to be $(n-1) \sum_i y_i$. Hence, we get the index of inequality given in \citet[equation~4]{stark_etal_2024}:
\begin{equation}\label{gini-sbf}
I^{SBF} = \left[\sum_{i=1}^{n} RD_i \right] \frac{1}{(n-1)\sum_{i=1}^{n} y_i}. 
\end{equation}
Once again, this can be written as a product of three factors, as in $I^{BC}$ or $I^{MSJ}$, but the factors are different. The first factor is the sum, rather than average, of relative deprivation. The second factor is the reciprocal of the total, rather than average, income in the network; the third factor is the reciprocal of $n-1$, rather than $1/2$.

\section{Key Properties of the Two Types of Measures}\label{sec:relationship}
Proposition~1 in \citet{stark_etal_2024} presents $4$ important properties of the relative deprivation-based measure of inequality for networks, $I^{SBF}$: it is bounded by the unit interval; it takes the minimum value of $0$ if income is divided equally between all the individuals; it takes the maximum value of $1$ if the network structure is a star network with the individual at the center receiving all the income (while all others have zero income); in a complete network, $I^{SBF} = G/(n-1) $, i.e. the measure of inequality is a rescaled value of the Gini coefficient \citep[page~352]{stark_etal_2024}.

Do the total experience-based measures of inequality for networks, $I^{BC}$ and $I^{MSJ}$, satisfy these properties? \citet{MAMUNURU2025106799} show that $I^{MSJ}$ is unbounded in terms of sample size: its value lies in the interval $[0,n/2]$, where $n$ is the number of individuals in the network. This is in contrast to $I^{SBF}$ which is bounded by the unit interval.  \citet{MAMUNURU2025106799} also show that $I^{MSJ}$ has similar properties relating to its maximum and minimum values: for a fixed total income and nodes, $I^{MSJ}$ attains the maximum value of $n/2$ for a star network organization with all income obtained by the individual at the center; $I^{MSJ}$ attains the minimum value of $0$ when income is equally distributed \citep[page~7]{MAMUNURU2025106799}.

In this section, I complement the discussion in \citet{MAMUNURU2025106799} about the properties of the total experience-based measures of inequality for networks by presenting three additional properties of $I^{BC}$ and $I^{MSJ}$. First, I show that $I^{BC}=I^{MSJ}$, for two special network structures: a complete network and a star network.\footnote{The result about complete networks is noted in \citet[page~9]{MAMUNURU2025106799}. I offer a slightly different proof.} For other network structures, these two network-based Gini coefficients have different values in general. Second, I show that for a complete network, the two social network-based reformulations of the Gini coefficients are bounded by the unit interval. For other network structures, that cannot be guaranteed. Third, I show that for a complete network, $I^{BC}=I^{MSJ}$ are rescaled versions of the standard Gini coefficient.

The strategy of proving equivalence, or otherwise, of the two total experience-based measures of inequality rely on comparing the first term in (\ref{gini-bowles}) and (\ref{gini-mjs}) because the other two terms are identical. For a network that is neither complete nor a star network, the terms being summed up to compute the first term in (\ref{gini-bowles}) and (\ref{gini-mjs}) are generally different. That is why $I^{BC} \neq I^{MSJ}$ for such networks. But for a complete network and for a star network, the first term in (\ref{gini-bowles}) and (\ref{gini-mjs}) coincide.

\subsection{Complete network}
\begin{lemma}\label{lem:bc_msj_cnet}
For a complete income network, where each individual is directly connected with all other individuals in the network, $I^{MSJ}=I^{BC}$.
\end{lemma}
\begin{proof}
For a complete network, using the definition of the neighborhood average difference in (\ref{eq:delta-defn}), we have
\[
\Delta_i = \frac{1}{n-1} \sum_{j=1}^{n} |y_j - y_i|, \quad i=1, 2, \ldots, n.
\]
Hence, the first term in (\ref{gini-mjs}) is
\begin{equation}\label{bc_msj_cnet}
\frac{1}{n} \sum_{i=1}^{n} \Delta_i = \frac{1}{n (n-1)} \sum_{i=1}^{n} \sum_{j=1}^{n} |y_j - y_i| = \frac{2}{n (n-1)} \sum_{i=1}^{n-1} \sum_{j = i+1}^{n} |y_j - y_i| = \frac{1}{\frac{n(n-1)}{2}} \sum_{i=1}^{n-1} \sum_{j = i+1}^{n} |y_j - y_i|
\end{equation}
where the penultimate step follows because $|y_j - y_i|=|y_i - y_j|$ and the expression in the last step is the first term in (\ref{gini-bowles}). This establishes the result.
\end{proof}

In a complete network each individual is connected to every other individual. Hence, every individual has the same number of connections and differ with each other only in terms of income levels. Thus, we can relabel them in increasing order of income levels without loss of generality.
\begin{theorem}\label{thm:kstar}
Consider a society of $n \geq 3$ individuals connected in a complete income network with income levels $y_i$, indexed by $i$ in nondecreasing order: $y_{i+1} \geq y_i$, for $i=1, 2, \ldots , n$. For this network, we have:
\begin{enumerate}
\item  $0 \leq I^{MSJ} = I^{BC} \leq 1$;
\item $I^{MSJ}/G=I^{BC}/G = 1/\left( 1 - (1/n)\right)$.
\end{enumerate}
\end{theorem}
The proof is available in Appendix~\ref{app:proof-kstar}. Here let us discuss the implications of the results. The first part of Theorem~\ref{thm:kstar} shows that $I^{MSJ}$ and $I^{BC}$ have the desirable property of being bounded by the unit interval if we restrict our attention to complete networks. For other network structures, this is not necessarily true. In contrast to this, $I^{SBF}$ is always bounded by the unit interval, irrespective of the network structure of society.

The second result in Theorem~\ref{thm:kstar} shows that, in a complete network,  $I^{MSJ}$ and $I^{BC}$ are rescaled versions of the standard Gini coefficient, much like $I^{SBF}$. But there is an interesting difference between the two sets of inequality measures because the scaling factors are different. For the latter, we have, according to argument in \citet[pages~351--352]{stark_etal_2024}, $I^{SBF}=G/(n-1)$; for the former, we have $I^{MSJ}=I^{BC} = G/\left( 1 - (1/n)\right)$. Thus, when $n \to \infty$ (number of individuals in the network becomes very large), $I^{SBF} \to 0$ (because $G$ is always bounded by the unit interval), but $I^{MSJ}=I^{BC} \to G$.\footnote{In the proof in \citet[pages~351--352]{stark_etal_2024}, the formula for the Gini coefficient a different normalization than the standard one I have used in (\ref{gini-std}). If we use (\ref{gini-std}), then we get $I^{SBF}=G/(n - 2 + (1/n))$. The limiting argument that $G \to 0$ as $n \to \infty$ remains unchanged.}   

\subsection{Star network}
\begin{lemma}\label{lem:bc-msj-snet}
Consider a star network with any individual $i$, with $1 \leq i \leq n$, at the center and individuals $j \neq i$ in the periphery. For this network, we have $I^{MSJ}=I^{BC}$.
\end{lemma}
\begin{proof}
For individuals indexed by $j \neq i$, i.e. those who are in the periphery, using the definition of neighborhood average difference in (\ref{eq:delta-defn}), we have
\[
\Delta_j = |y_j - y_i|,
\]
and for the $i$-th individual in the center, we have
\[
\Delta_i = \frac{1}{n-1} \sum_{j=1}^{n} |y_j - y_i|. 
\]
Hence, the first term in (\ref{gini-mjs}) is
\begin{equation}\label{bc-mjs-snet}
\frac{1}{n} \sum_{i=1}^{n} \Delta_i = \frac{1}{n} \left[ \sum_{j=1}^{n} |y_j - y_i| + \frac{1}{n-1} \sum_{j=1}^{n} |y_j - y_i|\right] = \frac{1}{n} \left[ \frac{n}{n-1} \sum_{j=1}^{n} |y_j - y_i|\right] =  \frac{1}{n-1} \left[ \sum_{j=1}^{n} |y_j - y_i|\right]
\end{equation}
where the expression in the last step is the first term in (\ref{gini-bowles}) for a star network with the $i$-th individual at the center. This establishes the result.
\end{proof}

To see that for network structures other than a complete network, $I^{BC}$ and $I^{MSJ}$ can be larger than $1$, we can simply refer to the fact noted by \citet{MAMUNURU2025106799} that $I^{MSJ}$ lies in the interval $[0,n/2]$, where $n$ is the number of individuals in the network. We can be more specific and give a simple counter example of a star network to understand precisely when $I^{MSJ}$ becomes larger than $1$.
\begin{proposition}\label{prop-msj-greater-1}
Consider a star network organization of a society of $n>3$ persons where the richest person is at the central node. If the richest person's income level, $y_n$, is more than $(3n-2)/(n^2-3n+2)$ times the sum of the income of all the other persons in this society, then $I^{MSJ} = I^{BC} > 1$.
\end{proposition}
\begin{proof}
Since the richest person is at the central node, computing the expression in (\ref{bc-mjs-snet}) with $i=n$, we see that if $y_n> 2\bar{y} + (n-1)^{-1} \sum_{i=1}^{n-1}y_i$, then $I^{BC} \left( 1 \right)>1$. With a little algebraic manipulation, this can be equivalently written as: if $y_n>((3n-2)/(n^2-3n+2))\sum_{i=1}^{n-1}y_i$, then $I^{BC}>1$. Since, as we saw in Lemma~\ref{lem:bc-msj-snet}, $I^{MSJ} = I^{BC} $ for a star network, the result is established. \end{proof}

Proposition~\ref{prop-msj-greater-1} shows that if the distribution of income is very skewed at the very top and society has a star network organization with the richest person at the center, then $I^{MSJ} = I^{BC} > 1$. Using this proposition, one can construct the following example of a $4$ person society with the following income levels, $1,2,4,12$, and with the richest person at the central node of a star network. In this case, $I^{MSJ} = I^{BC} =1.018$. 

\subsection{A network that is neither complete nor a star}
As an example of a case when the two measures, $I^{MSJ}, I^{BC} $, produce different values, consider a network where the adjacency matrix $A$ in (\ref{adj-matrix}) has $a_{12}=a_{21}=a_{jj}=0$, for $j=1, 2, \ldots, n$, and all other elements of $A$ are $1$. That is, individuals $1$ and $2$ are not directly connected to each other; all other network connections that obtain in a complete network are present. In this case,
\begin{equation}\label{bc-gen}
I^{BC} = \frac{1}{\frac{n(n-1)}{n} - 1} \left[ \sum_{i=3}^{n} |y_1 - y_i| + \sum_{i=3}^{n} |y_2 - y_i|  + \sum_{j=3}^{n-1} \sum_{i=j+1}^{n} |y_j - y_i|\right] \frac{1}{2 \bar{y}}.
\end{equation}
To compute $I^{MSJ}$, we need to compute the neighborhood average differences. Using (\ref{eq:delta-defn}), we get
\begin{align*}
\Delta_1 & = \frac{1}{n-2} \sum_{i=3}^{n} |y_1 - y_i| \\
\Delta_2 & = \frac{1}{n-2} \sum_{i=3}^{n} |y_2 - y_i| \\
\Delta_j & = \frac{1}{n-1} \sum_{i=1}^{n} |y_j - y_i|, \quad j=3, 4, \ldots, n. 
\end{align*}
Thus,
\begin{equation}\label{msj-gen}
I^{MSJ} = \frac{1}{n} \left[ \frac{1}{n-2} \sum_{i=3}^{n} |y_1 - y_i| + \frac{1}{n-2} \sum_{i=3}^{n} |y_2 - y_i|  + \frac{1}{n-1}  \sum_{j=3}^{n} \sum_{i=1}^{n} |y_j - y_i|\right] \frac{1}{2 \bar{y}}.
\end{equation}
In general, (\ref{bc-gen}) and (\ref{msj-gen}) are not equal. For network structures that are neither complete nor star networks, $I^{BC}$ and $I^{MSJ}$ are \textit{generally different}.

\section{Segregation Monotonicity and Network Inequality}\label{sec:star-k}

\subsection{Definitions}
I start by precisely defining two important concepts: segregation monotonicity and a level-$k$ star network. These two concepts are used to present the main results of this paper later in this section.

\begin{definition}[Segregation monotonicity]
Let $Y = \left( y_1, y_2, \ldots, y_n \right)$ denote an income distribution with at least one income level strictly positive; let $\mathcal{G}$ denote a network formed by these $n$ individuals. A network inequality measure $\mathcal{I} \left(Y,\mathcal{G} \right)$ satisfies segregation monotonicity if $\mathcal{I} \left(Y,\mathcal{G}' \right) \geq \mathcal{I} \left(Y,\mathcal{G} \right)$ whenever $\mathcal{G}'$ is more segregated than $\mathcal{G}$.
\end{definition}

That is, a network inequality measure satisfies segregation monotonicity if, holding the distribution of income fixed, its value does not decrease when the network undergoes a transformation that increases social segregation. In this paper, I consider a particular form of increasing segregation represented by a movement from what I call a level-$k$ star network to a level-$(k-1)$ star network. Under this transformation, relatively poorer individuals become increasingly isolated from one another, while their remaining social comparisons become increasingly concentrated among richer individuals. 

%
%

\begin{definition}\label{def:k-star-network}
Consider $n$ individuals in a social network. For $1 \leq k \leq n-1$, a level-$k$ star network obtains if the $n$ individuals are indexed by $i =1, 2, \ldots, n$, in nondecreasing order of income, $y_1 \leq y_2 \ldots \leq y_{n-1} \leq y_n$ and have the following pattern of connections: 
\begin{enumerate}
\item The poorest $n-k$ individuals, indexed by $i \in \left\lbrace 1, 2, \ldots, n-k\right\rbrace$, are isolated from each other.
\item The $k$ richest individuals, indexed by $i \in \left\lbrace n-k+1, n-k+2, \ldots, n\right\rbrace$, are connected to each other and to the poorest $n-k$ individuals.
\end{enumerate}
\end{definition}

\begin{figure}[htbp]
  \centering
  
  \begin{subfigure}[b]{0.31\textwidth}
    \centering
    \includegraphics[width=\textwidth]{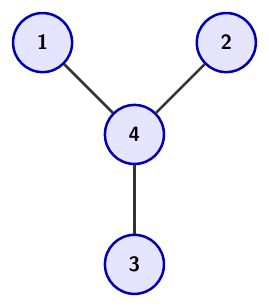}
    \caption{Level-1 Star}
    \label{fig:star-level1}
  \end{subfigure}
  \hfill 
  \begin{subfigure}[b]{0.31\textwidth}
    \centering
    \includegraphics[width=\textwidth]{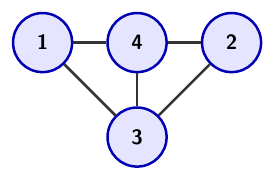}
    \caption{Level-2 Star}
    \label{fig:star-level2}
  \end{subfigure}
  \hfill
  \begin{subfigure}[b]{0.31\textwidth}
    \centering
    \includegraphics[width=\textwidth]{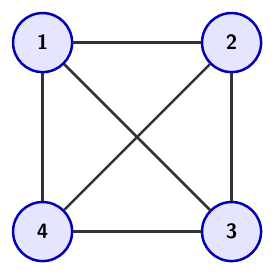}
    \caption{Level-3 Star}
    \label{fig:star-level3}
  \end{subfigure}

  \caption{Comparison of star networks across hierarchy levels.}
  \label{fig:star-networks-comparison}
\end{figure}

Let me flesh out the details of Definition~\ref{def:k-star-network} by considering various possibilities. A level-$1$ network in definition~\ref{def:k-star-network} is a standard star network with the richest person in the center, who is connected to all the other $(n-1)$ persons, while the $(n-1)$ poorer persons are isolated from each other. Thus, the richest person has the largest number of connections; she is the best connected individual in this society. This is a two class society, where the privileged class has only one member. An example of this network for a $4$ person society is given in Figure~\ref{fig:star-level1}.

A level-$2$ network in definition~\ref{def:k-star-network} represents a 2-class society, where the two richest persons belong to the privileged class. They are in the center (and connected to all the others) of the social network, while the other $(n-2)$ poorer persons are all isolated from each other (and all belong to the second class). This is a two class society, where the privileged class has $2$ members. An example of this network for a $4$ person society is given in Figure~\ref{fig:star-level2}.

For any $k$ between $1$ and $n-1$, we can construct a level$-k$ star network using the following recursion: the level-$k$ star network is the union of the level-$(k-1)$ star network and new connections (edges) formed between the $(n-k+1)$-th individual with her poorer counterparts. This is a two class society, where the privileged class has $k$ members. Finally, for $k=(n-1)$, we get the complete (fully connected) network, where every person is connected to every other person. Here, everyone is equal in terms of the number of connections. Thus, everyone belongs to the same class. An example of this network for a $4$ person society is given in Figure~\ref{fig:star-level3}.

\subsection{Intuition}
The intuition underlying the concept of a level-$k$ star network comes from my engagement with previous literature. For instance, in emphasizing the merits of $I^{BC}$, \citet{BOWLES2020108789} discuss a three-person example with wealth levels of $10,4,3$. If we compare the situation of these three individuals in a complete network with another one in which the richest person is in the center with the other two persons isolated from each other, we would see that $I^{BC}$ increases from $0.41$ to $0.57$ with no change in the wealth of the three individuals \citep{BOWLES2020108789}. 

This example captures an important intuitive understanding of `experienced inequality' (which $I^{BC}$ is meant to measure). When, in the three-person example, the complete network is transformed into a star network with the richest individual at the center, the average income level of comparators (only one person in this example) of the two poorer individuals rise. Hence, inequality experienced by them increase. There is of course no change in the experience of inequality for the richest person because her comparators do not change. Since $I^{BC}$ is an aggregate of experienced inequality across all individuals, it rises.

The star network structure of this example captures the intuitive idea that richer individuals or households might be advantaged in terms of having larger number of connections (with all the benefits that come from better or larger social connections). That is, in addition to having larger income or wealth, richer individuals or households might also have better or more social connections (and therefore more social capital). Relatively poorer individuals have fewer social connections, especially with other poor individuals; they are relatively more isolated than richer individuals. When such isolation increases, society becomes more segregated. The intuition is that, holding the income distribution fixed, greater isolation of poorer individuals changes the composition of their socially relevant comparisons toward richer individuals. 


Definition~\ref{def:k-star-network} formalizes this intuition for a general $n$-person network in a specific manner where, by construction, I keep individual income levels fixed but allow relatively poorer individuals to becomes increasingly more isolated, i.e. society becomes increasingly segregated. Thus, the \textit{specific type of network transformation} captured by a transition from a level-$k$ to a level-$(k-1)$ network allows me to compare the two approaches to measuring inequality in terms of how well they capture our intuition about increasing inequality as networks become more segregated keeping income distribution unchanged.

\subsection{Implementation}
To compute the measures of inequality in a level-$k$ star network, we need to understand the specific structures of the matrices $A, D$ and $M$ for this network architecture. In this case, the adjacency matrix, $A$, defined in (\ref{adj-matrix}) becomes, 
\begin{equation}\label{adj-matrix-kstar}
a_{ij} = 1 \text{ if and only if } i \neq j \text{ and at least one of } i,j \text{ belongs to the set of richest } k \text{ individuals}. 
\end{equation}
This follows from the requirement in Definition~\ref{def:k-star-network} that the poorest $n-k$ are isolated from each other but connected to the richer $k$ individual (who are themselves connected with each other). For a level-$k$ star network, the absolute difference matrix, $D$, defined in (\ref{diff-matrix}), becomes 
\begin{equation}\label{diff-matrix-kstar}
d_{ij} = d_{ji} = 
\begin{cases}
y_j - y_i  \quad \text{ if } j > i  \\
0 \quad \text{ otherwise},
\end{cases}
\end{equation}
where I can remove the absolute values because the income levels are nondecreasing in the index. Thus, $D$, remains symmetric with zero principal diagonal but the elements no longer have absolute values. Finally, for a level-$k$ star network, the positive difference matrix, $P$, defined in (\ref{pos-diff-matrix}), becomes lower triangular,
\begin{equation}\label{pos-diff-matrix-kstar}
p_{ij} = 
\begin{cases}
y_i - y_j  \quad \text{ if } i > j  \\
0 \quad \text{ otherwise},
\end{cases}
\end{equation}
once again, because the income levels are nondecreasing in the index.

In the next two sub-sections, I present the main results of this paper. These results examine whether alternative measures of inequality satisfy segregation monotonicity with respect to the specific network transformation represented by a movement from a level-$k$ star to a level-$(k-1)$ star network. For presenting these results, I will use the following notation: for $k=1, 2, \ldots, n-1$, $I^{SBF} \left(k \right)$ and $I^{MSJ} \left(k \right)$ denote, respectively, $I^{SBF}$  and $I^{MSJ}$ in a level-$k$ star network.

\subsection{Relative deprivation-based approach}

\begin{theorem}[Monotonicity]\label{thm:sbf-snet}
Consider an income network with $n$ individuals. Keeping individual income levels fixed, if society moves from a level-$k$ to a level-$(k-1)$ star network organization, then $I^{SBF} \left(k - 1 \right) \geq I^{SBF} \left(k \right)$, for $k=2, 3, \ldots, n-1$.
\end{theorem}
The proof of the theorem is given in Appendix~\ref{app:proof-sbf-snet}. The theorem says that, with individual income levels remaining fixed, if society becomes more segregated in the sense of moving from a level-$k$ to a level-$(k-1)$ star network organization, then the value of  $I^{SBF}$ registers an increase. Moreover, this is true for every $k=2, \ldots, n-1$. As an example, if society moves from a complete network organization (level-$(n-1)$ star network) to an organization where the two poorest persons are isolated from each other keeping all other connections intact, the value of $I^{SBF}$ increases. At the other extreme, if society moves from a level-$2$ star network (where the two richest individuals are connected to all others, while the poorest $(n-2)$ individuals are isolated from each other) to a level-$1$ star network (where only the richest individual is connected to all others, while the poorest $(n-1)$ individuals are isolated from each other), then $I^{SBF}$ registers an increase. Our intuition suggests that, in general, a movement from level-$k$ to a level-$(k-1)$ star network increases isolation of the relatively poor. If increased segregation of this form is regarded as increasing inequality, then a measure satisfying segregation monotonicity should increase under this transformation. Theorem~\ref{thm:sbf-snet} shows that $I^{SBF}$ satisfies segregation monotonicity with respect to this transformation.

\subsection{Total experience-based approach}
I now construct a counterexample to show that $I^{BC}$ and $I^{MSJ}$ do not generally satisfy segregation monotonicity with respect to this transformation. In particular, I show that if income levels are not too skewed at the very top end of the income distribution, moving from a level-$2$ to a level-$1$ star network organization of society keeping individual income level unchanged, i.e. from a relatively less to a relatively more segregated society, can reduce the value of $I^{BC}$ and $I^{MSJ}$. 

\begin{theorem}[Lack of monotonicity]\label{thm:counter-int}
Suppose $y_1 \leq y_2 \leq \cdots \leq y_{n-1} \leq y_n$ denotes the income levels of the $n$ individuals in a society. 
\begin{enumerate}
\item The condition 
\[
y_{n-1} \leq y_n \leq \left( \frac{n}{n-1} \right) y_{n-1} - \frac{1}{(n-1)(n-2)} \sum_{i=1}^{n-2} y_i 
\]
is necessary and sufficient for $I^{BC} \left( 2 \right) \geq I^{BC} \left( 1 \right)$.

\item The condition 
\[
y_{n-1} \leq y_n \leq \left( \frac{n^2+n-6}{n^2-n-2} \right) y_{n-1} - \frac{2}{n^2-n-2} \sum_{i=1}^{n-2} y_i 
\]
is necessary and sufficient for $I^{MSJ} \left( 2 \right) \geq  I^{MSJ} \left( 1 \right)$. 

\end{enumerate}
\end{theorem}

The proof of the theorem is given in Appendix~\ref{app:proof-thm-cint}. Here, let us discuss the intuition. The theorem derives us an upper bound on the highest income level as a function of the other income levels in society for which the total experienced-based measures violate segregation monotonicity. Thus, they show that as long as income levels are not too skewed at the very top end of the income distribution, moving from a level-$2$ to a level-$1$ star network organization of society while individual income levels remain fixed, i.e. from a relatively less segregated to a relatively more segregated society, can reduce the value of $I^{BC}$ and $I^{MSJ}$. This is in contrast to the relative deprivation-based approach, where an increase in segregation, keeping individual income levels fixed, registers an increase in $I^{SBF}$, irrespective of income distribution.

Using the results of Theorem~\ref{thm:counter-int}, we can construct the following example of a $5$ person society, where income levels are: $1,2,4,50,55.96$. In this case, implementing (\ref{gini-bowles}), we get $I^{BC} \left(4 \right) = 0.6990$, $I^{BC} \left(3 \right) = 0.7742$, $I^{BC} \left(2 \right) = 0.9796$ and $I^{BC} \left(1 \right) = 0.9231$. Thus, keeping individual income levels unchanged, as the society moves from the complete network to the most segregated society, the value of $I^{BC}$ first rises  and then falls. The movement from the $k=2$ to the $k=1$ case leads to a \textit{decline} in $I^{BC}$ from $0.980$ to $0.923$, directly contradicting the intuitive expectation that greater network segregation increases experienced inequality. 

For the same example $5$ person society, let us now implement (\ref{gini-mjs}). Now we get $I^{MSJ} \left(4 \right) = 0.6990$, $I^{MSJ} \left(3 \right) = 0.7746$, $I^{MSJ} \left(2 \right) = 1.022$ and $I^{MSJ} \left(1 \right) = 0.9231$. Once again, keeping individual income levels unchanged, as the society moves from the complete network to the most segregated society, the value of $I^{MSJ}$ first rises  and then falls. The movement from the $k=2$ to the $k=1$ case leads to a \textit{decline} in $I^{MSJ}$ from $1.022$ to $0.923$. Thus, this example does not merely show that the measures fail to increase strictly; it shows that they can decrease under an increase in segregation.

\section{Conclusion}\label{sec:conclusion}
Measures of inequality that are relevant for networks are functions of income distribution \textit{and} network structure. In this paper, I have proposed segregation monotonicity as a criterion by which such measures of inequality can be evaluated and compared. If a measure of inequality satisfies segregation monotonicity then its value should weakly rise as a network becomes more segregated in some well-defined sense with income distribution remaining unchanged. 

The central contribution of this paper is to compare two recently developed approaches to constructing measures of inequality in networks using the criterion of segregation monotonicity. The first, total experience-based, approach developed by \citet{BOWLES2020108789} and extended by \citet{MAMUNURU2025106799}, measures inequality by aggregating absolute income differences across all network neighbors. The second, relative deprivation-based, approach developed by \citet{stark_etal_2024}, by contrast, focuses on upward comparisons: individuals experience deprivation only relative to directly connected individuals with higher incomes. 

To compare the two approaches, I introduced the class of level-$k$ star networks. Moving from a higher-level to a lower-level network increases the isolation of relatively poorer individuals and concentrates their comparisons among relatively richer individuals. This provides a precise way of representing an increase in social segregation without changing any individual's income.

For this class of network transformations, the relative deprivation-based measure satisfies segregation monotonicity. As the network becomes more segregated, the value of the measure weakly increases. The reason is straightforward: poorer individuals compare themselves with a progressively more affluent set of neighbors, while the income distribution itself remains fixed. The relative deprivation-based measure therefore tracks the intuition that increased isolation of poorer individuals and increased concentration of their social comparisons among richer individuals constitute an increase in inequality.

The total experience-based measures do not possess this property in general. The counterexamples developed in this paper show that, when income is not too strongly concentrated at the very top, a movement from a less segregated to a more segregated level-$k$ star network can reduce the values of both the Bowles-Carlin and Mamunuru-Shrivastava-Jayadev measures. Thus, greater segregation need not be registered as greater inequality by measures based on the aggregate of absolute differences among connected individuals. 

This divergence arises because the two approaches respond differently to the composition of socially relevant comparisons. The total experience-based approach treats upward and downward differences symmetrically, whereas the relative deprivation-based approach gives special significance to upward comparisons. The comparison of these approaches shows that the choice between them is not merely a matter of selecting different formulas for measuring the same object. Rather, the measures embody different conceptions of how social structure, in the form of a network, informs the meaning of inequality. 

An additional set of results presented in this paper concerns the properties of the total experience-based measures. I showed that the Bowles-Carlin measure and the Mamunuru-Shrivastava-Jayadev measure coincide in two important network structures: complete networks and star networks. In a complete network, the two measures are also rescaled versions of the standard Gini coefficient and are bounded by the unit interval. These results clarify the relationship between the standard Gini coefficient and its network-based reformulations. At the same time, the results also highlight an important limitation of the total experience-based measures: outside complete networks their values need not be bounded by one and can depend substantially on the particular architecture of social connections.

The results in this paper have a broader implication for the measurement of inequality in networks. Once social structure is allowed to determine `who compares with whom,' the question of which comparisons should matter becomes inseparable from the normative and conceptual definition of inequality itself. A measure based on absolute differences answers the question of how unequal individuals are relative to their social contacts in both directions. A relative deprivation measure answers a different question: how much disadvantage individuals experience relative to better-off members of their social environment. Neither perspective can therefore be evaluated solely by the conventional axioms of inequality measurement; their suitability also depends on the type of social comparison that the measure is intended to represent.

The results also suggest that network segregation should receive greater attention in the analysis of inequality. Two societies can have exactly the same distribution of income but differ substantially in the social environments through which economic differences are observed and experienced. If poorer individuals become increasingly isolated from one another while remaining connected to richer individuals, the pattern of experienced economic differences changes even though conventional distributional inequality remains unchanged. A satisfactory theory of inequality in social networks must therefore specify not only how income differences are aggregated but also how changes in the structure of social comparison should affect the resulting measure.

The analysis in this paper is limited to a particular and deliberately simple class of segregated networks. This simplicity is useful because it isolates the role of network architecture from changes in the income distribution. An important direction for future research is to investigate whether the divergence identified here between total experience-based and relative deprivation-based indexes of inequality extends to broader classes of networks and to characterize the network transformations under which alternative measures of inequality satisfy segregation monotonicity. Another important direction is empirical. The increasing availability of data containing both economic characteristics and social connections makes it possible to investigate whether different network-based measures produce systematically different assessments of inequality in actual societies.

The main conclusion of the paper is therefore that measuring inequality in a network requires more than restricting the set of interpersonal comparisons. It also requires a theory of which comparisons matter and how changes in the structure of those comparisons should be interpreted. The relative deprivation-based and total experience-based approaches offer two distinct answers to that question. The comparison developed here shows that this distinction has substantive consequences: when social networks become increasingly segregated while incomes remain unchanged, the measures can move in opposite directions. Understanding this divergence is perhaps important for the development of a coherent theory of inequality in socially structured populations.

\appendix

\counterwithin{equation}{section}

\section{Proofs of theorems}

\subsection{Proof of Theorem~\ref{thm:kstar}}\label{app:proof-kstar}
\textit{Part~1.} The lower bound holds because each term in (\ref{gini-bowles}) and (\ref{gini-mjs}) are nonnegative. To see the upper bound, let us write out the sum for the right most expression in (\ref{bc_msj_cnet}):
\[
I^{BC} = \left[ \frac{2 \left[ (n-1)y_n + (n-2)y_{n-1} + \cdots + 3 y_4 + 2 y_3 + y_2 \right] - n(n-1)\bar{y}}{\frac{n(n-1)}{2}} \right] \frac{1}{2 \bar{y}}.
\]
To show that this ratio is weakly less than $1$, we need to show that
\[
2\left[ (n-1)y_n + (n-2)y_{n-1} + \cdots + 2 y_3 + y_2 \right] - n(n-1)\bar{y} \leq n(n-1)\bar{y}.
\]
Moving $n(n-1)\bar{y}$ to the right hand side and using $n \bar{y}=(y_1 + y_2 + \cdots + y_n)$, we see that the above will be established if we have
\[
\left[ (n-1)y_n + (n-2)y_{n-1} + \cdots + 2 y_3 + y_2 \right] \leq (n-1) (y_1 + y_2 + \cdots + y_n).
\]
Collecting terms, we see that the above holds because
\[
y_{n-1} + 2 y_{n-2} + \cdots + (n-2) y_2 + (n-1) y_1 \geq 0,
\]
where the last step holds since $y_i \geq 0$. This establishes that  $I^{BC} \leq 1$. Moreover, we have seen in Lemma~\ref{lem:bc_msj_cnet} that for a complete network, $I^{BC}=I^{MSJ}$. Hence, we get: $I^{MSJ} = I^{BC} \leq 1$.

\textit{Part~2.} The expression for the Gini coefficient can be written as
\[
G = \frac{1}{2n^2 \bar{y}} \sum_{j=1}^{n} \sum_{i=1}^{n} |y_j - y_i| = \frac{2}{2n^2 \bar{y}} \sum_{j=1}^{n-1} \sum_{i=j+1}^{n} |y_j - y_i|,
\]
because $|y_j - y_i|=|y_i - y_j|$. Using the expression for $I^{BC}$ in (\ref{gini-bowles}) and dividing with the above version of the Gini coefficient, we get
\[
\frac{I^{BC}}{G} = \frac{n^2}{n^2 - n} = \frac{1}{1 - (1/n)}
\]
which completes the proof.

\subsection{Proof of Theorem~\ref{thm:sbf-snet}}\label{app:proof-sbf-snet}
The strategy of the proof is the following: in the first step, I will compute the sum of the relative deprivation measure for the $n$ individuals when society is organized as a level-$k$ star network (I will identify the level of the network with a superscript); in the second step, I will compute the sum of the relative deprivation measure for the $n$ individuals when society is organized as a level-$(k-1)$ star network; in the third step, I will take the difference between the sums in the second and first steps and show that the result is nonnegative.

\textit{Step~1.} Let us compute the relative deprivation measure for the $n$ individuals when society is organized as a level-$k$ star network. For individuals indexed by $i = 1, 2, \ldots, n-k$, the relative deprivation is governed primarily by the structure of the adjacency matrix $A$ in (\ref{adj-matrix-kstar}). In particular, we have
\begin{equation}
RD_i^{(k)} = \frac{1}{k} \sum_{j=1}^{k} \left( y_{n+1-j} - y_i \right) = \left( \frac{1}{k} \sum_{j=1}^{k} y_{n+1-j} \right) - y_i.
\end{equation}
On the other hand, for individuals indexed by $i = n-k+1, \ldots, n-1, n$, the relative deprivation is primarily driven by the structure of the positive difference matrix, $P$, in (\ref{pos-diff-matrix-kstar}). In particular, we have
\begin{align*}
RD_{n-k+1}^{(k)} & = \left( \frac{1}{n-1} \sum_{j=1}^{k-1} y_{n+1-j} \right) - \left( \frac{k-1}{n-1} \right) y_{n-k+1}, \\
RD_{n-k+2}^{(k)} & = \left( \frac{1}{n-1} \sum_{j=1}^{k-2} y_{n+1-j} \right) - \left( \frac{k-2}{n-1} \right) y_{n-k+2}, \\
 & \vdots \\
RD_{n-1}^{(k)} & =   \left( \frac{1}{n-1} y_{n} \right) -  y_{n-1}, \\
RD_n^{(k)} & = 0
\end{align*}

\textit{Step~2.} Let us compute the relative deprivation measure for the $n$ individuals when society is organized as a level-$(k-1)$ star network. For individuals indexed by $i = 1, 2, \ldots, n-k+1$, the relative deprivation is governed primarily by the structure of the adjacency matrix $A$ in (\ref{adj-matrix-kstar}). In particular, we have
\begin{equation}
RD_i^{(k-1)} = \frac{1}{k-1} \sum_{j=1}^{k-1} \left( y_{n+1-j} - y_i \right) = \left( \frac{1}{k-1} \sum_{j=1}^{k-1} y_{n+1-j} \right) - y_i.
\end{equation}
On the other hand, for individuals indexed by $i = n-k+2, \ldots, n-1, n$, the relative deprivation is identical to what it is in a level-$k$ star network. In particular, we have
\begin{align*}
RD_{n-k+2}^{(k-1)} & = \left( \frac{1}{n-1} \sum_{j=1}^{k-2} y_{n+1-j} \right) - \left( \frac{k-2}{n-1} \right) y_{n-k+2}, \\
 & \vdots \\
RD_{n-1}^{(k-1)} & =   \left( \frac{1}{n-1} y_{n} \right) -  y_{n-1}, \\
RD_n^{(k-1)} & = 0
\end{align*}

\textit{Step~3.} Consider individuals indexed by $i$, where $n-k+2 \leq i \leq n$. For these individuals, the relative deprivation measure does not change between the level-$k$ and level-$(k-1)$ star networks.

Now, consider individuals indexed by $i$, where $1 \leq i \leq n-k$. For these individuals,
\begin{equation}
RD_{i}^{(k-1)} - RD_{i}^{(k)} = \left( \frac{1}{k-1} \sum_{j=1}^{k-1} y_{n+1-j} \right) - \left( \frac{1}{k} \sum_{j=1}^{k} y_{n+1-j} \right),
\end{equation}
which is the difference between the arithmetic mean income of the top $(k-1)$ and the arithmetic mean income of the top $k$ incomes. Since incomes are arranged in nondecreasing order, the first is no smaller than the second. Hence, the difference is nonnegative.

Finally consider \textit{the} individual indexed by $i=n-k+1$. For this individual,
\[
RD_{n-k+1}^{(k-1)} = \left( \frac{1}{k-1} \sum_{j=1}^{k-1} y_{n+1-j} \right) - y_{n-k+1}, \quad RD_{n-k+1}^{(k)} = \left( \frac{1}{n-1} \sum_{j=1}^{k-1} y_{n+1-j} \right) - \left( \frac{k-1}{n-1}\right) y_{n-k+1}.
\]
Thus,
\begin{align*}
RD_{n-k+1}^{(k-1)} - RD_{n-k+1}^{(k)} \geq 0
\end{align*}
if and only if
\begin{align*}
\left( \frac{1}{k-1} - \frac{1}{n-1}\right) \sum_{j=1}^{k-1} y_{n-j+1} \geq \left( 1 - \frac{k-1}{n-1}\right) y_{n-k+1}
\end{align*}
if and only if
\begin{align*}
\left( \frac{n-k}{(k-1)(n-1)}\right) \sum_{j=1}^{k-1} y_{n-j+1} \geq \left( \frac{n-k}{n-1}\right) y_{n-k+1}
\end{align*}
if and only if
\begin{equation}
y_{n-k+1} \leq \frac{1}{k-1} \sum_{j=1}^{k-1} y_{n-j+1}
\end{equation}
where the last inequality is true because the right hand side is the average income of the richest $(k-1)$ individuals and the left hand side is the income level of the richest $(k-1)$-th  richest individual and, according to definition~\ref{def:k-star-network}, incomes are nondecreasing in the index $i$.

Thus, term by term, the relative deprivation measure is higher or equal for each individual when the society moves to a level-$(k-1)$ star network from a level-$k$ star network keeping individual income levels unchanged. Hence, the numerator in (\ref{gini-sbf}) is higher for the level-$(k-1)$ star network than for a level-$k$ star network. Since the denominators are the same for both (because individual income levels remain unchanged), this implies that $I^{SBF} \left( k-1 \right) \geq I^{SBF} \left( k \right)$.

\subsection{Proof of Theorem~\ref{thm:counter-int}}\label{app:proof-thm-cint}
\textit{Part~1.} Let $F^{BC}_{1}$ and $F^{BC}_{2}$ denote the first term in (\ref{gini-bowles}) for a level-$1$ and level-$2$ network, respectively. Then, we have
\begin{align}
F^{BC}_{1} & = \frac{1}{n-1} \sum_{i=1}^{n-1} \left( y_n - y_i \right) = y_n - \left( \frac{1}{n-1} \right)y_{n-1} - \frac{1}{n-1} \sum_{i=1}^{n-2} y_i, \label{eq:bc_f_n1} \\
F^{BC}_{2} & = \frac{1}{2n-3} \left[ \sum_{i=1}^{n-1} \left( y_n - y_i \right) + \sum_{i=1}^{n-2} \left( y_{n-1} - y_i \right) \right]  = \left( \frac{n-1}{2n-3} \right) y_n - \left( \frac{n-3}{2n-3} \right)y_{n-1} - \frac{2}{2n-3} \sum_{i=1}^{n-2} y_i. \label{eq:bc_f_n2}
\end{align}
We want to find the condition under which (\ref{eq:bc_f_n2}) is larger than (\ref{eq:bc_f_n1}), that is
\[
\left( \frac{n-1}{2n-3} \right) y_n - \left( \frac{n-3}{2n-3} \right)y_{n-1} - \frac{2}{2n-3} \sum_{i=1}^{n-2} y_i \geq y_n - \left( \frac{1}{n-1} \right)y_{n-1} - \frac{1}{n-1} \sum_{i=1}^{n-2} y_i
\]
which is equivalent to establishing
\[
\left(\frac{n-3}{2n-3} + \frac{1}{n-1}\right)y_{n-1} + \left(\frac{1}{n-1} - \frac{2}{2n-3}\right) \sum_{i=1}^{n-2} y_i \geq \left(1 - \frac{n-1}{2n-3}\right)y_{n}
\]
which is equivalent to establishing
\[
\left(\frac{n}{n-1}\right)y_{n-1} - \frac{1}{(n-1)(n-2)} \sum_{i=1}^{n-2} y_i \leq y_n,
\]
the upper bound in the first part of the theorem. Since, $y_{n-1} \leq y_n$ by assumption, this establishes the lower bound and completes the first part of the theorem.

\textit{Part~2.} Let $F^{MSJ}_{1}$ and $F^{MSJ}_{2}$ denote the first terms in the respective expressions of $I^{MSJ}$ in (\ref{gini-mjs}) for a level-$1$ and level-$2$ network. Then, 
\begin{align}
F^{MSJ}_{1} & = \frac{1}{n-1} \left[ \left( n-1 \right) y_n - y_{n-1} - \sum_{i=1}^{n-2} y_i \right] \label{eq:mjs_f_n1}\\
F^{MSJ}_{2} & = \frac{1}{n(n-1)} \left[ \left( \frac{n^2-n+2}{2} \right) y_n + \left( \frac{n^2-n-6}{2} \right) y_{n-1} - \left(n+1 \right) \sum_{i=1}^{n-2} y_i \right]. \label{eq:mjs_f_n2}
\end{align}
To see how we get (\ref{eq:mjs_f_n1}), recall that the first terms in (\ref{gini-bowles}) and (\ref{gini-mjs}) are identical for a star network. Since a level-$1$ start network is a standard star network, $F^{MSJ}_{1}=F^{BC}_{1}$ from (\ref{eq:bc_f_n1}), which has been rearranged to give (\ref{eq:mjs_f_n1}). To see how we get (\ref{eq:mjs_f_n2}), we need to use the definition of the neighborhood average difference in (\ref{eq:delta-defn}). Using this definition, for a level-$2$ star network, we have
\begin{align*}
\Delta_i & = \frac{1}{2} \left( y_n - y_i + y_{n-1} - y_i\right), \quad i=1, 2, \ldots, n-2; \\
\Delta_{n-1} & = \frac{1}{n-1} \left(y_{n-1} - y_1 + \cdots + y_{n-1} - y_{n-2} + y_{n} - y_{n-1} \right); \\ 
\Delta_{n} & = \frac{1}{n-1} \left(y_{n} - y_1 + \cdots + y_{n} - y_{n-2} + y_{n} - y_{n-1} \right).
\end{align*}
Using these expressions, we can compute 
\begin{align*}
F^{MSJ}_{2} & = \frac{1}{n} \sum_{i=1}^{n} \Delta_i \\
& = \frac{1}{2n} \left\lbrace (n-2) \left( y_{n} + y_{n-1} \right) - 2 \sum_{i=1}^{n-2} y_i \right\rbrace  + \frac{1}{n(n-1)} \left\lbrace y_n + (n-3) y_{n-1}  - \sum_{i=1}^{n-2} y_i \right\rbrace  \\
 & \quad + \frac{1}{n(n-1)} \left\lbrace (n-1) y_n -  y_{n-1}  - \sum_{i=1}^{n-2} y_i \right\rbrace \\
 & = \frac{1}{n} \left\lbrace \left( \frac{n^2 -n+2}{2(n-1)} \right) y_n + \left(\frac{n^2 -n-6}{2(n-1)} \right) y_{n-1}  - \left(\frac{n+1}{n-1} \right) \sum_{i=1}^{n-2} y_i \right\rbrace,
\end{align*}
which is the expression in (\ref{eq:mjs_f_n2}).

Let us now find the condition under which $(n-1)F^{MSJ}_{2} \geq (n-1)F^{MSJ}_{1}$. Using the expressions in (\ref{eq:mjs_f_n2}) and (\ref{eq:mjs_f_n1}), we have the $(n-1)F^{MSJ}_{2}\geq (n-1)F^{MSJ}_{1}$ if and only if
\begin{align*}
\left( \frac{n^2 -n+2}{2n} \right) y_n + \left(\frac{n^2 -n-6}{2n} \right) y_{n-1}  - \left(\frac{n+1}{n} \right) \sum_{i=1}^{n-2} y_i \geq  (n-1) y_n - y_{n-1} - \sum_{i=1}^{n-2} y_i
\end{align*}
which is equivalent to
\[
\left(\frac{n^2 +n-6}{2n} \right) y_{n-1}  - \frac{1}{n}  \sum_{i=1}^{n-2} y_i \geq  \left( \frac{n^2 -n-2}{2n} \right) y_n
\]
which is equivalent to
\[
\left(\frac{n^2 +n-6}{n^2 -n-2} \right) y_{n-1}  - \frac{2}{n^2 -n-2}  \sum_{i=1}^{n-2} y_i \geq y_n,
\]
which is the upper bound in the second part of the theorem. Since, $y_{n-1} \leq y_n$ by assumption, this gives the lower bound and establishes the second part of the theorem.

\bibliographystyle{apalike} 
\bibliography{sgini-refs}    

\begin{thebibliography}{}

\bibitem[Bowles and Carlin, 2020]{BOWLES2020108789}
Bowles, S. and Carlin, W. (2020).
\newblock Inequality as experienced difference: A reformulation of the {Gini}
  coefficient.
\newblock {\em Economics Letters}, 186:108789.

\bibitem[Bowles and Carlin, 2024]{bowles_carlin_2024}
Bowles, S. and Carlin, W. (2024).
\newblock Axioms and intuitions about societal inequality: What does the {Gini}
  coefficient measure?
\newblock {\em Journal of Income Distribution}, 32(3-4):45--54.

\bibitem[Chiang, 2011]{chiang_2011}
Chiang, Y.-S. (2011).
\newblock Judgment of distributional inequality in networks.
\newblock {\em Social Networks}, 33:342--349.

\bibitem[Chiang, 2015]{chiang_2015}
Chiang, Y.-S. (2015).
\newblock Inequality measures perform differently in global and local
  assessments: An exploratory computational experiment.
\newblock {\em Physica A}, 437:1--11.

\bibitem[Deaton, 1997]{deaton_1997}
Deaton, A. (1997).
\newblock {\em The Analysis of Household Surveys: A Microeconometric Approach
  to Development Policy}.
\newblock World Bank, Washington, D.C.

\bibitem[DiMaggio and Garip, 2012]{dimaggio_garip_2012}
DiMaggio, P. and Garip, F. (2012).
\newblock Network effects and social inequality.
\newblock {\em Annual Review of Sociology}, 38:93--118.

\bibitem[Festinger, 1954]{festinger_1954}
Festinger, L. (1954).
\newblock A theory of social comparison processes.
\newblock {\em Human Relations}, 7(2):117--140.

\bibitem[Jasso, 1979]{jasso_1979}
Jasso, G. (1979).
\newblock {On Gini's mean difference and Gini's index of concentration}.
\newblock {\em American Sociological Review}, 44(5):867--870.

\bibitem[Lambert, 2001]{lambert_2001}
Lambert, P.~J. (2001).
\newblock {\em The Distribution and Redistribution of Income}.
\newblock Manchester University Press, New York, NY, {Third} edition.

\bibitem[Mamunuru et~al., 2025]{MAMUNURU2025106799}
Mamunuru, S.~M., Shrivastava, A., and Jayadev, A. (2025).
\newblock Social networks and experienced inequality.
\newblock {\em Journal of Economic Behavior \& Organization}, 229:106799.

\bibitem[Ray, 1998]{ray_1998}
Ray, D. (1998).
\newblock {\em Development Economics}.
\newblock Princeton University Press, Princeton, NJ.

\bibitem[Ray, 2024]{ray_2024}
Ray, D. (2024).
\newblock Notes on ``{Notes on the Gini} coefficient''.
\newblock {\em Journal of Income Distribution}, 32(3-4):40--44.

\bibitem[Sen, 1997]{sen_1997}
Sen, A. (1997).
\newblock {\em On Economic Inequality. Expanded edition with a substantial
  annexe by James E. Foster and Amartya Sen}.
\newblock Oxford University Press, New York, NY.

\bibitem[Sethi, 2024]{sethi_2024}
Sethi, R. (2024).
\newblock Notes on the {Gini} coefficient.
\newblock {\em Journal of Income Distribution}, 32(3-4):35--39.

\bibitem[Stark et~al., 2024]{stark_etal_2024}
Stark, O., Bielawski, J., and Falniowski, F. (2024).
\newblock Measuring income inequality in social networks.
\newblock {\em The Journal of Economic Inequality}, 22:333--356.

\bibitem[T{\'o}th et~al., 2021]{Toth2021}
T{\'o}th, G., Wachs, J., Di~Clemente, R., Jakobi, {\'A}., S{\'a}gv{\'a}ri, B.,
  Kert{\'e}sz, J., and Lengyel, B. (2021).
\newblock Inequality is rising where social network segregation interacts with
  urban topology.
\newblock {\em Nature Communications}, 12:1143.

\bibitem[Yitzhaki, 1979]{yitzhaki_1979}
Yitzhaki, S. (1979).
\newblock Relative deprivation and the {Gini} coefficient.
\newblock {\em The Quarterly Journal of Economics}, 93(2):321--324.

\end{thebibliography}

\end{document}